\documentclass[lettersize,journal]{IEEEtran}
\usepackage{amsmath,amsfonts,amssymb}
\usepackage{algorithmic}
\usepackage{algorithm}
\usepackage{amsthm}
\usepackage{array}
\usepackage[caption=false]{subfig}
\usepackage{textcomp}
\usepackage{stfloats}
\usepackage{url}
\usepackage{verbatim}
\usepackage{graphicx}
\usepackage{color}
\usepackage{soul}

\usepackage{cite}
\usepackage[thinc]{esdiff}
\DeclareMathOperator{\Exp}{\mathbb{E}}
\newtheorem{theorem}{Theorem}
\newtheorem{lemma}[theorem]{Lemma}

\begin{document}

\title{On the Age of Status Updates in Unreliable Multi-Source M/G/1 Queueing Systems}

\author{Muthukrishnan\,Senthil\,Kumar,~\IEEEmembership{Member,~IEEE},
        Aresh\,Dadlani,~\IEEEmembership{Senior Member,~IEEE},
        Masoumeh\,Moradian,
        Ahmad\,Khonsari,
        Theodoros\,A.\,Tsiftsis,~\IEEEmembership{Senior Member,~IEEE}
\thanks{This work was supported by the National Natural Science Foundation of China under Grant 62071202.}
\thanks{Muthukrishnan Senthil Kumar is with the Department of Applied Mathematics and Computational Sciences, PSG College of Technology, Coimbatore 641004, India (e-mail: msk.amcs@psgtech.ac.in).}%
\thanks{Aresh Dadlani is with the Faculty of Arts - Interdisciplinary Studies,~University of Alberta, Edmonton, AB T6G 2E7, Canada, and also with the~Department of Electrical and Computer Engineering, Nazarbayev University, Astana 010000,~Kazakhstan\,(e-mail:\,dadlani@ieee.org).}
\thanks{Masoumeh Moradian is with the School of Computer Science, Institute for Research in Fundamental Sciences, Tehran, Iran (e-mail: mmoradian@ipm.ir).}
\thanks{Ahmad Khonsari is with the Department of Electrical and Computer~Engineering, University of Tehran, and also with the Institute for Research in Fundamental Sciences, Tehran, Iran (e-mail: a\_khonsari@ut.ac.ir).}
\thanks{Theodoros\,A.\,Tsiftsis is with the School of Intelligent Systems Science and Engineering, Jinan University, Zhuhai Campus, Zhuhai 519070, China (e-mail: theo\_tsiftsis@jnu.edu.cn).}}%

\maketitle

\vspace{-0.5em}
\begin{abstract}
\fontdimen2\font=0.52ex
The timeliness of status message delivery in communications networks is subjective to time-varying wireless channel transmissions. In this paper, we investigate the age of information (AoI) of~each source in a multi-source M/G/1 queueing update system with~active server failures. In particular, we adopt the method of supplementary~variables to derive a closed-form expression for the average AoI in terms of system parameters, where the server repair time follows a general distribution and the service time of packets generated by independent sources is a general random variable. Numerical results are provided to validate the effectiveness of the proposed packet serving policy under different parametric settings.\vspace{-1.2em}
\end{abstract}

\begin{IEEEkeywords}
Age of information, multi-source M/G/1 queueing model, server failure, supplementary variable, real-time systems.\vspace{-0.6em}
\end{IEEEkeywords}

\section{Introduction}
\label{intro}
\fontdimen2\font=0.50ex
\IEEEPARstart{A}{ge} of Information (AoI) has emerged as an instrumental performance metric to quantify information freshness at~the receiver in status update systems facilitated by advances in 6G wireless networks. Unlike network delay, the measured value~of a monitored communication process is included in a status~update packet, together with a timestamp indicating the moment the~data was generated at the source. In upcoming cyber-physical control systems, AoI is essentially used to fully characterize latency and is related to the time elapsed since the latest received update was generated at any given time \cite{Yates2021, Andrea2021, Moradian2021, Miridakis2022}. 

Various queueing models have recently been adopted to~investigate AoI and other age-related metrics in low-latency~communication systems with shared medium and random user behavior. The average AoI (AAoI) for a single-source~first-come, first-served (FCFS) M/M/1 queue was first analyzed in \cite{Kaul2012},~and then extended to a multi-source setup in \cite{Yates2019}. To address stochastic~update service time distributions, the authors of \cite{Najm2018} derived the AAoI in closed form for an M/G/1/1~preemptive queue with multiple streams. Furthermore, three~approximate AAoI expressions for a multi-source FCFS M/G/1 queue were derived in \cite{Moltafet2020}. In \cite{Moltafet2021},~the authors used the stochastic hybrid systems paradigm to~derive the moment generating function (MGF) of AoI under the self-preemptive and non-preemptive packet management policies. In \cite{Elmagid2022}, the distributional properties of AoI in terms of the MGF for several energy-harvesting queueing disciplines were presented.

In remote health monitoring and underwater sensor networks where multiple transmitters send status updates over wireless links, the impact of service disruptions on information freshness remains largely unexplored. The AAoI and the peak AoI for a Ber/G/1 vacation queue was first studied in \cite{Tripathi2019}. More recently, the M/G/1 vacation queue was studied under different scheduling schemes \cite{Xu2022}. Apart from~vacations, random server breakdown (with repair) is another form of interruption inherent~in emerging wireless networks with unreliable medium. Nevertheless, to our best knowledge, no queueing-theoretic assessment on status age under active server breakdowns and repairs has yet been reported.

Aiming to fill this research gap, the main contributions of~this letter are: (i) derivation of the steady-state distribution and the probability generating function (pgf) of AoI in a single-source M/G/1 queue with unreliable wireless transmission using the supplementary variable technique; (ii) closed-form AoI expression for the multi-source M/G/1 queue with active server breakdowns and general repair time; (iii) numerical evaluation of the AAoI with respect to packet arrival and service failure rates.
\vspace{-0.6em}
\begin{figure}[!t]
	\centering
	\includegraphics[width=0.4\textwidth]{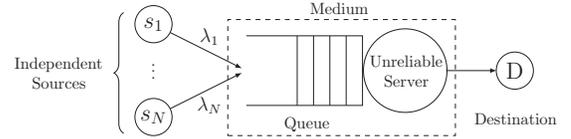}
	\vspace{-0.6em}
	\caption{The unreliable M/G/1 queueing model with multiple sources.}
	\label{fig1}
	\vspace{-0.8em}
\end{figure}


\section{System Model}
\label{sec2}
\fontdimen2\font=0.50ex
In this section, we first describe the M/G/1 queueing model with server failures before formally defining the associated AoI.\vspace{-0.8em}

\subsection{M/G/1 Queueing Model with Unreliable Server}
\label{sec2.1}
\fontdimen2\font=0.50ex
Consider a single-server system in which each independent data source in the set $\mathcal{S} \!=\! \{s_k\}$, where $k \!=\! 1,2,\ldots,N$, generates packets following a Poisson process with rate $\lambda_k \!>\! 0$, which~are then immediately sent to the unreliable server as shown in \figurename~{\ref{fig1}}. The processing (service) time of each packet is an i.i.d. random variable (RV) with distribution function $F_{S_k\!}(t)$,~where $S_k$ denotes the general service time RV of a packet generated~by source $s_k$. The corresponding Laplace-Stieltjes transform (LST) is given as $S_k^*(s) \!=\!\! \int_{0}^{\infty}\! e^{\!-st} \text{d}F_{S_k\!}(t)$ and the $i$-th order moment is $\beta^{(k,i)} \!=\! (-1)^i \frac{\text{d}^i}{\text{d} s^i} S_k^*(s)|_{s \rightarrow 0^+}$. The server, which is prone~to breakdowns while serving packets, has an exponentially distributed lifetime with mean $1/\alpha$. Upon breakdown, the server undergoes a repair process that enables the packet that was being served at the time of breakdown to resume its remaining service. Let $R_k$ denote the general repair time of the server while processing a packet from $s_k$ and $F_{R_k\!}(t)$ be its distribution function. The corresponding LST and $i$-th order moment are $R_k^*(s) \!=\!\! \int_{0}^{\infty}\! e^{-st} \text{d} F_{R_k\!}(t)$~and $\gamma^{(k,i)} \!=\! (-1)^i \frac{\text{d}^i}{\text{d} s^i} R_k^*(s)|_{s \rightarrow 0^+}$,~respectively. Moreover, the arrival process of packets, server idle time, packet service time,~and server repair time are assumed to be mutually independent.\vspace{-1em}

\subsection{Average AoI Definition}
\label{sec2.2}
\fontdimen2\font=0.50ex
As defined in \cite{Moltafet2020}, let $t_{k,i}$ denote the time instant at which the $i$-th status update packet of source $s_k$ was generated, and $t'_{k,i}$ be the instant at which it arrives at the destination $\mathrm{D}$. Thus, the most recent packet of $s_k$ received at time instant $\tau$ can be indexed as:\vspace{-0.3em}
\begin{equation}
	\mathcal{N}_k(\tau) \!=\! \max\{i' | t'_{k,i'} \leq \tau\}.
	\label{eq1}
	\vspace{-0.3em}
\end{equation} 
The timestamp of the recently received packet of $s_k$ is $\mathcal{U}_k(\tau) \!=\! t_{k,\mathcal{N}_k(\tau)}$. Consequently, the AoI of $s_k$ at the destination is~characterized as the random process $\Delta_k(t) \!=\! t - \mathcal{U}_k(t)$. The time~average AoI of $s_k$ at destination $\mathrm{D}$ in interval $(0,\tau)$ is defined as \cite{Moltafet2020}:\vspace{-0.3em}
\begin{equation}
	\Delta_{\tau,k} \!\triangleq\! \frac{1}{\tau}\int_{0}^{\tau} \Delta_k(t) dt,
	\label{eq2}
	\vspace{-0.3em}
\end{equation}
such that the area under $\Delta_k(t)$ is the sum of the disjoint areas $B_{k,i}$, where $i \!=\!1,2,\ldots,\mathcal{N}_k(\tau)$, as shown in \figurename~{\ref{fig2}}. Assuming that the random process $\{B_{k,i}\}$ is mean ergodic, the AAoI of~$s_k$, defined as $\Delta_k \triangleq \lim_{\tau \rightarrow \infty} \Delta_{\tau,k}$, converges to the stochastic average \cite{Moltafet2020}. Thus, we have $\Delta_k \!=\! \lambda_k \Exp[B_{k,i}]$.
\begin{figure}[!t]
	\centering
	\includegraphics[width=0.9\columnwidth]{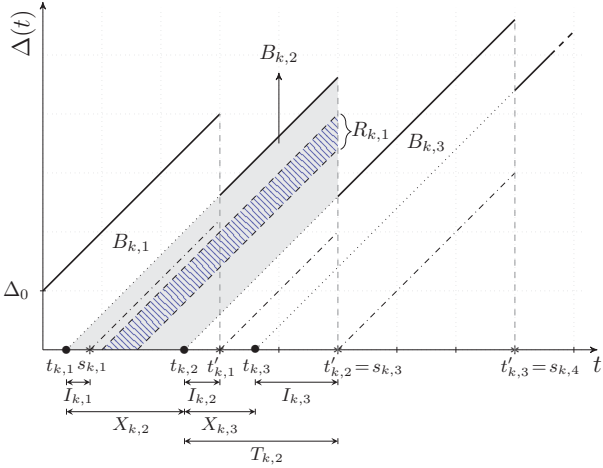}
	\vspace{-0.8em}
	\caption{AoI evolution of $s_k$ packets with respect to time $t$. The blue strip highlights the server repair time while processing packet 1 from $s_k$.}
	\label{fig2}
	\vspace{-0.6em}
\end{figure}

The RV $X_{k,i} \!=\! t_{k,i} \!-\! t_{k,i-1}$ depicted in \figurename{~\ref{fig2}} marks the~$i$-th inter-arrival time of $s_k$, $s_{k,i}$ denotes the instant at which the~transmission of the $(k,i)$-th packet commences, $I_{k,i} \!=\! s_{k,i} \!-\! t_{k,i}$ is the $i$-th idle time of the server for source $s_k$, and the system~sojourn time is denoted by RV $T_{k,i} \!=\! t'_{k,i} \!-\! t_{k,i}$. In a classical queueing model, the transmission of the $(k, i-1)$-th packet completes at the same time as the transmission of the $(k,i)$-th packet begins only if the packet generated at instant $t_{k,i}$ lies within~the interval $(s_{k,i-1},t'_{k,i-1})$. During the transmission period, however, the server is subject to active breakdowns.~The repair process begins as soon as the server fails. In such a case, the packet being~served before the server breakdown resumes its remaining service immediately after the server is repaired.~Hence, $T_{k,i}$ is the~system time of packet $(k,i)$, i.e., the sum of the packet waiting time $(W_{k,i})$ and the service time $(S_{k,i})$ which itself, includes the~possible repair time when the system fails $(R_{k,i})$. Undertaking the approach in \cite{Moltafet2020}, the AAoI of source $s_k$ is determined to be:\vspace{-0.2em}
\begin{align}
	\Delta_k &\!=\! \lambda_k \Exp[B_{k,i}] \!=\! \lambda_k \left(\frac{1}{2} \Exp\left[(X_{k,i} + T_{k,i})^2\right] \!-\! \frac{1}{2} \Exp\left[X^2_{k,i}\right]\right) \nonumber \\
	&\!=\! \lambda_k \left(\frac{\Exp[X^2_{k,i}]}{2} + \Exp\left[X_{k,i} (W_{k,i} + S_{k,i})\right]\right).
	\label{eq3}
	\vspace{-0.5em}
\end{align}

\section{Single-Source Unreliable M/G/1 Queue: Steady-State Distribution}
\label{sec3}
\fontdimen2\font=0.50ex
On breaking down while processing a update packet, the~server in an M/G/1 queue instantly undergoes repair, which has a~generally distributed time. Without loss of generality, we assume~that the single source, $s_1$, generates packets at rate $\lambda_1 \!>\! 0$. We use $S^{\circ}(t)$ and $R^{\circ}(t)$ to denote the remaining packet service and server repair times, respectively, at $t \geq 0$. If $M(t)$ is the total of packets in the queue at time $t$, then the server state can be defined as:\vspace{-0.3em}
\begin{equation}
  C(t) =
  \begin{cases}
    0, & \text{ if the server is \textit{idle}}, \vspace{-0.25em}\\
    1, & \text{ if the server is \textit{serving a packet}}, \vspace{-0.25em}\\
    2, & \text{ if the server \textit{fails while serving a packet}}, 
  \end{cases}
  \label{eq4}
  \vspace{-0.2em}
\end{equation}
For generally distributed service and repair times, the following continuous-time Markov chain captures the system state at $t$:\vspace{-0.3em}
\begin{equation}
	K(t) = \{\left(C(t), M(t), S^{\circ}(t), R^{\circ}(t)\right); t \geq 0\}.
	\label{eq5}
	\vspace{-0.3em}
\end{equation}
Based on \eqref{eq5}, the state probabilities are defined as follows, where $m,x,y \geq 0$:\vspace{-0.3em}
\begin{equation}
  \begin{cases}
    p_0(t) = \Pr[C(t)=0, M(t) = 0], \vspace{-0.1em}\\
    p_1^m(x,t)\,\text{d}x = \Pr[C(t) \!=\! 1, M(t) \!=\! m, x <\! S^{\circ}(t) \!\leq\! x + \text{d}x], \vspace{-0.1em}\\
    p_2^m(x,y,t) \text{d}y = \Pr[C(t) \!=\! 2, M(t) \!=\! m, S^{\circ}(t) = x, \\
     \qquad\qquad\qquad\qquad\qquad\qquad\qquad y < R^{\circ}(t) \leq y + \text{d}y], 
  \end{cases}
  \label{eq6}
  \vspace{-0.2em}
\end{equation}
where $p_0(t)$ denotes the probability that the server is idle at time $t$, $p_1^m(x,t)\,\text{d}x$ refers to the joint probability that the server is busy transmitting a packet during the remaining service time $(x, x+\text{d}x)$ while there are $m$ packets in the queue at time $t$, $p_2^m(x,y,t)\,\text{d}y$ refers to the joint probability that there are $m$ queued packets, the remaining service time is $x$, and the failed server is fixed within the remaining repair time $(y, y \!+\!\text{d}y)$ while serving a packet.

We now apply the method of supplementary variables \cite{kleinrock} to obtain the following balance equations from \eqref{eq6} as $t \rightarrow \infty$:\vspace{-0.3em}
\begin{equation}
  \begin{cases}
    \!\lambda_1 p_0 = p_1^0(0),\vspace{0.3em} \\
    \!\diff{p_1^0(x)}{x} \!=\! (\lambda_1 \!+\! \alpha) p_1^0(x) \!-\! \left(\lambda_1 p_0 \!+\! p_1^1(0)\right)\! f_s(x) \!-\! p_2^0(x,0), \vspace{0.3em}\\
    \!\diff{p_1^{m\!}(x)\!}{x} \!=\! (\!\lambda_{1\!} \!+\! \alpha\!) p_1^{m\!}(x) \!-\!\! \lambda_1 p_1^{m \!-\! 1\!}(x) \!-\! p_1^{m \!+\! 1\!}(0) f_{\!s}(x) \!-\! p_2^{m\!}(x,\!0), \vspace{0.3em}\\
    \!\diffp{p_2^{m\!}(x,y)}{y} \!=\! \lambda_1 p_2^m(x,y) \!-\! \lambda_1 p_2^{m \!-\!1}(x,y) - \alpha p_1^m(x) f_r(y),
  \end{cases}
  \label{eq7}
  \vspace{-0.2em}
\end{equation}
where $f_s(x)$ and $f_r(x)$ are, respectively, the probability density functions of the service and repair times, and the normalization condition is given as:\vspace{-0.3em}
\begin{equation}
	p_0 + \sum_{m=0}^{\infty} \left[\int_{0}^{\infty}\! p_1^m(x)\,\text{d}x + \int_0^{\infty}\!\! \int_0^{\infty}\! p_2^m(x,y)\,\text{d}x\,\text{d}y\right] \!=\! 1.
	\label{eq8}
	\vspace{-0.3em}
\end{equation}

Denoting the LSTs of $p_1^m(x)$ and LST$[p_2^m(x,y)]$ by $\tilde{p}_1^{\,m}(\theta)$~and $\widetilde{\tilde{p}}_2^{\,m}(\theta,s)$, respectively, the marginal pdfs are:\vspace{-0.3em}
\begin{equation}
  \begin{cases}
    \tilde{p}_1(z,\theta) = \sum_{m=0}^{\infty} \tilde{p}_1^{\,m}(\theta) z^m,\vspace{0.3em} \\
    p_1(z,0) = \sum_{m=0}^{\infty} p_1^m(0) z^m, \vspace{0.3em}\\
    \widetilde{\tilde{p}}_2(z,\theta,s) = \sum_{m=0}^{\infty} \widetilde{\tilde{p}}_2^{\,m}(\theta,s) z^m, \vspace{0.3em}\\
    \tilde{p}_2(z,\theta,0) = \sum_{m=0}^{\infty} \tilde{p}_2^{\,m}(\theta,0) z^m,
  \end{cases}
  \label{eq9}
  \vspace{-0.3em}
\end{equation}
By applying LST on both sides of \eqref{eq7}, multiplying $z^m$ to~the resultant equations, summing over $m$, and through some algebraic manipulations, the pgfs of the queue size, $\mathcal{P}(z)$, and the system size, $\mathcal{Q}(z)$, are derived to be as follows, where $p_0 \!=\! 1 \!-\! \lambda_1 \beta^{(1,1)}(1 \!+\! \alpha \gamma^{(1,1)})$ and $\phi(s,z) \!=\! \left(s \!+\! \lambda_1 \!-\! \lambda_1 z \!+\! \alpha \!-\! \alpha R^*(s \!+\! \lambda_1 \!-\! \lambda_1 z)\right)$:\vspace{-0.3em}
\begin{equation}
  \begin{cases}
    \!\!\mathcal{P}(z) \!=\! p_0 \!+\! \tilde{p}_1(z,\!0) \!+\! \widetilde{\tilde{p}}_2(z,\!0,\!0) \!=\! \dfrac{\left(1 \!-\! S^*\!\left(\phi(0,\!z)\right)\right)p_0}{S^*\!\left(\phi(0,\!z)\right) \!-\! z},\vspace{0.3em} \\
    \!\!\mathcal{Q}(z) \!=\! p_0 \!+\! z \tilde{p}_1(z,\!0) \!+\! z \widetilde{\tilde{p}}_2(z,\!0,\!0) \!=\! \dfrac{S^*\!\left(\phi(0,\!z)\right)\!(1 \!-\! z)p_0}{S^*\!\left(\phi(0,\!z)\right) \!-\! z},
  \end{cases}
  \label{eq10}
  \vspace{-0.2em}
\end{equation}

To obtain the LST of sojourn time distribution of the system, we substitute $a \!=\! \lambda_1(1 -z)$ in $\mathcal{Q}(z)$ given in \eqref{eq10}:\vspace{-0.2em}
\begin{equation}
	W^*(a) = \dfrac{a \left(S^*\!\left(a + \alpha \left(1 - R^*(a)\right)\right)\right) p_0}{\left(a - \lambda_1\left(1 - S^*(a + \alpha \left(1 - R^*(a)\right))\right)\right)}.
	\label{eq11}
	\vspace{-0.2em}
\end{equation}

The system availability at time $t$, denoted by $A(t)$, determines the likelihood of when the unreliable server will be available for use. In steady-state, i.e., $P_a = \lim_{t\to\infty} A(t)$, we get:\vspace{-0.1em}
\begin{equation}
	P_a = \lim_{z\to 1} \{p_0 + \tilde{p}_1(z,0)\} \!=\! 1 - \lambda_1\beta^{(1,1)} \alpha \gamma^{(1,1)}.
	\label{eq11b}
	\vspace{-0.4em}
\end{equation}

\section{Multi-Source Unreliable M/G/1 Queue: AoI Analysis}
\label{sec4}
\fontdimen2\font=0.50ex
In this section, we derive the AAoI defined in \eqref{eq3} for the~unreliable M/G/1 queueing model with multiple sources. We begin our analysis of \eqref{eq3} by considering sources $s_1$ and $s_2$. For the first term in \eqref{eq3}, we have $\Exp[X^2_{1,i}] \!=\! 2/\lambda_1^2$ since the inter-arrival time of~$s_1$ follows an exponential distribution with $\lambda_1 \!>\! 0$. The second term in \eqref{eq3} can be expressed as follows:\vspace{-0.3em}
\begin{equation}
	\Exp[X_{1,i} (W_{1,i} + S_{1,i})] = \Exp[X_{1,i} W_{1,i}] + \Exp[X_{1,i}] \Exp[S_{1,i}].
	\label{eq12}
	\vspace{-0.2em}
\end{equation}
In \eqref{eq12}, the second term arises as the inter-arrival time and service time (including the repair time) are independent of each other. However, evaluation of the first term in \eqref{eq12} is computationally complex since the inter-arrival and waiting times are dependent. To find $\Exp[X_{\!1,i}W_{1,i}]$, we express $W_{1,i}$ in terms of two events, namely $A^b_{1,i}$ and $A^l_{1,i}$. The former represents the event when the inter-arrival time of packet $(1,i)$ is shorter than the system time of packet $(1, i \!-\! 1)$, whereas the latter is the complementary event. Under event $A^b_{1,i}$, the waiting time of packet $(1,i)$ comprises~of (i) the remaining time to complete processing packet $(1, i \!-\! 1)$, (ii) the sum of service times of $s_2$ packets that arrived during $X_{1,i}$ and must be served before packet $(1,i)$ according to the FCFS rule, and (iii) the sum of possible server repair times before serving packet $(1,i)$. Similarly, for $A^l_{1,i}$, the waiting time of packet $(1,i)$ includes (i) the remaining service time of $s_2$ packet under service at the instant packet $(1,i)$ arrives, (ii) the sum of service times of $s_2$ packets served prior to packet $(1,i)$ in FCFS order, and (iii) the total repair time before processing the $(1,i)$-th packet.

For event $A^b_{1,i}$, let $R^b_{1,i} \!=\! T_{1,i-1} \!-\! X_{1,i}$ be the residual~system time to complete serving packet $(1,i \!-\! 1)$. Also, let the total~service times of $s_2$ packets arriving during $X_{1,i}$ that need service before packet $(1,i)$ be $S^b_{1,i\!} \!=\! \!\sum_{i' \in \Omega^b_{2,i}}\!\! S_{2,i'}$, where $\Omega^b_{2,i}$ is the index set of queued packets from $s_2$. Let $F^b_{1,i} \!=\! \sum_{i' \in \Omega^b_{2,i}}\!\! R_{2,i'}$ be the total server repair time before serving the~$(1,i)$-th packet. In the same manner, we define $S^l_{1,i} \!=\! \sum_{i' \in \Omega^l_{2,i}}\!\! S_{2,i'}$ and $F^l_{1,i} \!=\! \sum_{i' \in \Omega^l_{2,i}}\!\! R_{2,i'}$ for event $A^l_{1,i}$, where $\Omega^l_{2,i}$ is the set of indices of queued packets generated by source $s_2$. Using these definitions, the waiting time of packet $(1,i)$ can be written as:\vspace{-0.3em}
\begin{equation}
  W_{1,i} =
  \begin{cases}
    S^b_{1,i} + F^b_{1,i} + R^b_{1,i}, & \text{ if }A^b_{1,i}\text{ occurs}, \vspace{-0.05em}\\
    S^l_{1,i} + F^l_{1,i} + R^l_{2,i}, & \text{ if }A^l_{1,i}\text{ occurs},
  \end{cases}
  \label{eq13}
  \vspace{-0.3em}
\end{equation}
where $R^l_{2,i}$ denotes the residual service time RV of the $s_2$ packet under service at the arrival of packet $(1,i)$ conditioned on event $A^l_{1,i}$. By taking the conditional expectation on \eqref{eq12}, we get:\vspace{-0.4em}
\begin{align}
	\Exp[X_{\!1,i} W_{\!1,i}] \!=\! &\left(\Exp[R^b_{1,i} X_{\!1,i}|A^b_{1,i}] \!+\! \Exp[(S^b_{1,i\!} \!+\! F^b_{1,i}) X_{\!1\!,i} | A^b_{1,i}]\right)p^b_{1,i} \nonumber \\
	&+\! \Exp[(S^l_{1,i} \!+\! F^l_{1,i} \!+\! R^l_{2,i}) X_{1,i} | A^l_{1,i}]\, p^l_{1,i},
	\label{eq14}
	\vspace{-0.4em}
\end{align}
where $p^b_{1,i}$ and $p^l_{1,i}$ are the probabilities of events $A^b_{1,i}$ and $A^l_{1,i}$, respectively, and are derived to be:\vspace{-0.4em}
\begin{equation}
    \begin{cases}
		p^b_{1,i} = 1 - p^l_{1,i}, \vspace{0.15em}\\
    	p^l_{1,i} = \dfrac{S^*\big(\lambda_1 \!+\! \alpha \left(1 - R^*(\lambda_1)\right)\!\big) \lambda_1 p_0}{\big(\lambda_1 \!-\! \lambda \left(1 \!-\! S^*\!\left(\lambda_1 \!+\! \alpha \left(1 - R^*(\lambda_1)\right)\right)\right)\!\big)},
  \end{cases}
  \label{eq15}
  \vspace{-0.4em}
\end{equation}
where $\lambda \!=\! \sum_k \lambda_k$, and the service and sojourn times of all packets are stochastically identical, i.e., $\forall i$ and $k \!=\! 1,2,\ldots,N$, we have $S_{k,i} \!=\! S$ and $T_{k,i} \!=\! T$. For the sake of presentation, we substitute $G^b_{1,i} \!=\! S^b_{1,i} \!+\! F^b_{1,i}$ and $G^l_{1,i} \!=\! S^l_{1,i} \!+\! F^l_{1,i}$ in \eqref{eq14}. From probability theory, the conditional probability density function of RVs $X_{1,i}$ and $T_{1,i-1}$, given the event $A^b_{1,i}$, can be expressed as:\vspace{-0.4em}
\begin{equation}
  f_{X_{1,i}, T_{1,i-1}|A^b_{1,i}\!}(s,t) \!=\!
  \begin{cases}
    \!\dfrac{\lambda_1 e^{-\lambda_1 s} f_{T_{1,i-1}}(t)}{p^b_{1,i}}, & \text{ if }s \leq t, \vspace{0.03em}\\
    \!0, & \text{ if }s > t.
  \end{cases}
  \label{eq16}
  \vspace{-0.3em}
\end{equation}
Using \eqref{eq16}, we derive the three conditional expressions of \eqref{eq14} in Lemmas~\ref{lemma1}-\ref{lemma3}, where $\forall k$, we assume $\rho_k \!=\! \lambda_k \beta^{(1,1)}(1 \!+\! \alpha \gamma^{(1,1)})$.
\begin{lemma}
	\label{lemma1}
	The closed-form expression for the first conditional expectation in \eqref{eq14} is given as:\vspace{-0.4em}
\begin{equation}
	 \Exp[R^b_{1\!,i} X_{\!1\!,i}|A^b_{1\!,i}] \!=\!\! \dfrac{1}{\lambda_1 p^b_{1,i}\!\!} \!\!\left(\!\Exp[W\!] \!+\! \Exp[S] \!-\!\! W^{*'\!}\!(\!\lambda_{1\!}) \!+\! \dfrac{2 W^{*\!}(\lambda_{1\!}) \!-\!\! 2}{\lambda_1}\!\right)\!.   
  	\label{eq17}
  	\vspace{-0.4em}
\end{equation}
\end{lemma}
\begin{proof}
	See Appendix~\ref{appA}.\vspace{-0.3em}
\end{proof}
\begin{lemma}
	\label{lemma2}
	The closed-form expression for the second conditional expectation in \eqref{eq14} is given as:\vspace{-0.4em}
\begin{equation}
	 \Exp[G^b_{1\!,i} X_{\!1\!,i}|A^b_{1\!,i}] \!=\!\! \dfrac{\rho_2}{p^b_{1,i}\!\!} \!\left(\!\!\dfrac{2}{\lambda_1^2} \!-\! W^{*''\!\!}(\lambda_{1\!}) \!+\!\! \dfrac{2 W^{*'\!}(\lambda_{1\!})}{\lambda_1} \!+\! \dfrac{2 W^{*\!}(\!\lambda_{1\!})}{\lambda_1^2}\!\right)  \!. 
  	\label{eq18}
  	\vspace{-0.4em}
\end{equation}
\end{lemma}
\begin{proof}
	See Appendix~\ref{appB}.\vspace{-0.3em}
\end{proof}
\begin{lemma}
	\label{lemma3}
	The closed-form expression for the third conditional expectation in \eqref{eq14} is given as:\vspace{-0.4em}
\begin{equation}
	 \Exp[\left(G^l_{1\!,i} \!+\! R^l_{2,i}\right) X_{\!1\!,i}|A^l_{1\!,i}] = \dfrac{\rho_2}{p^l_{1,i}} \!\left(\! W^{*''\!\!}(\lambda_{1}) \!-\!\! \dfrac{W^{*'\!}(\lambda_{1})}{\lambda_1}\!\right)\!. 
  	\label{eq19}
  	\vspace{-0.3em}
\end{equation}
\end{lemma}
\begin{proof}
	See Appendix~\ref{appC}.\vspace{-0.3em}
\end{proof}
\begin{figure}[!t]
	\centering
	{\label{fig3b}\includegraphics[width=0.518\columnwidth]{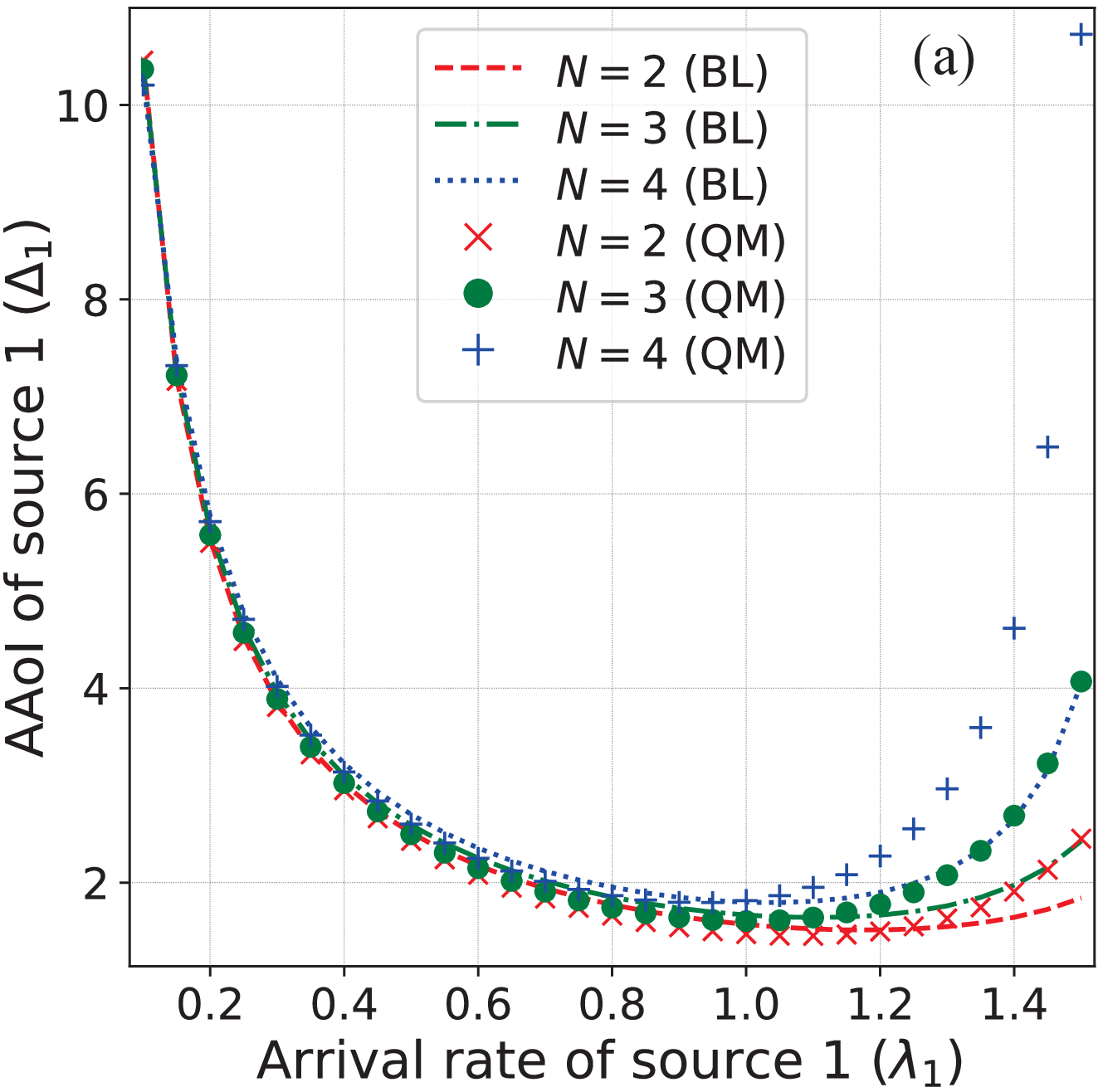}}
	{\label{fig3c}\includegraphics[width=0.463\columnwidth]{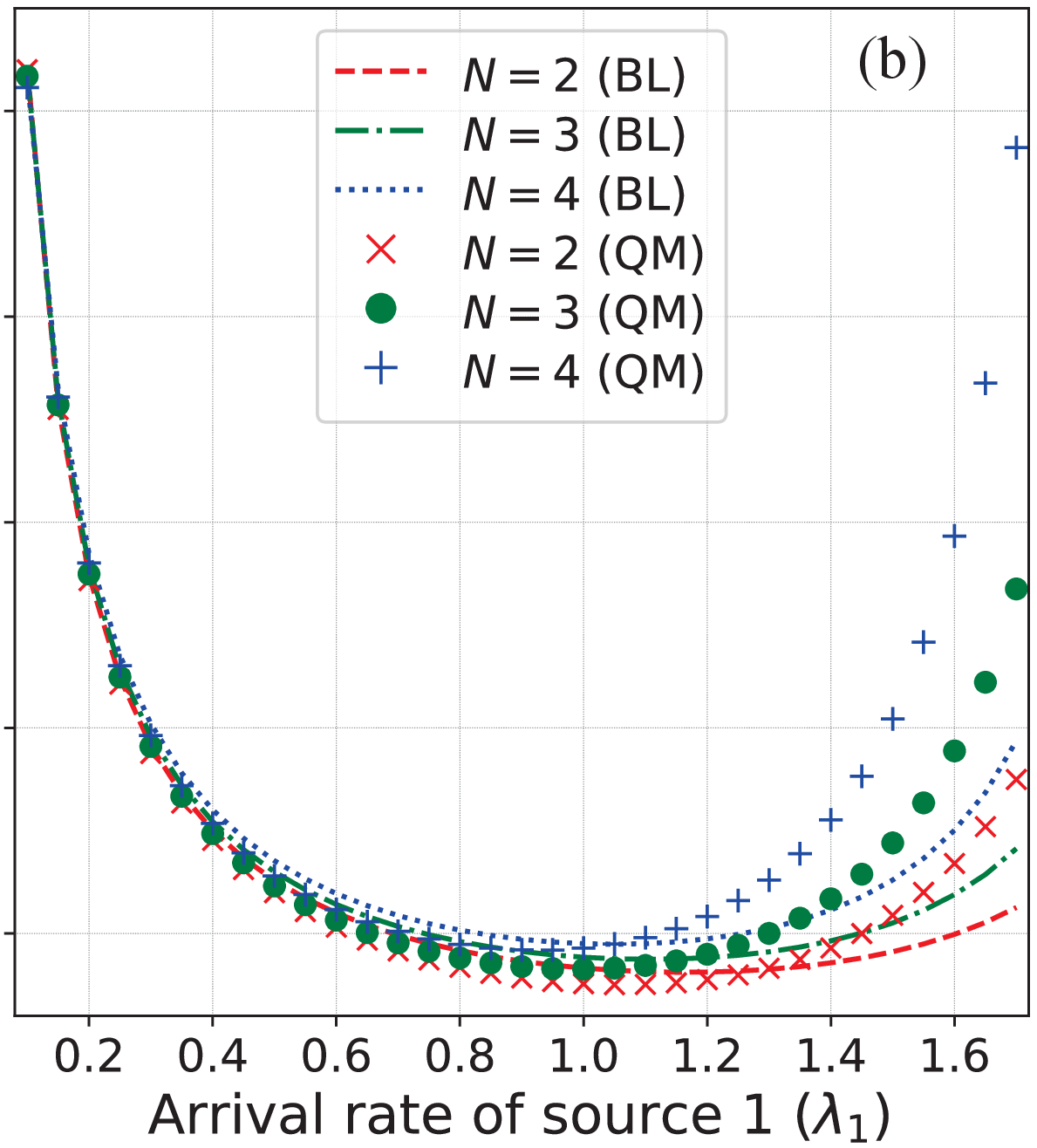}}
\vspace{-0.7em}
	\caption{The AAoI of source $s_1$ in terms of $\lambda_1$ for $N \!=\! \{2,3,4\}$ under (a) Erlang-$2$ and (b) hyper-exponential service time distributions with $p \!=\! 0.7$. Here, $\Exp[S]=0.5$, $\Exp[R]=0.3$, $\alpha = 0.1$, and $\lambda_i (i \neq 1) \!=\! 0.12$.}
	\vspace{-0.4em}
	\label{fig3}
\end{figure}

From \eqref{eq17}-\eqref{eq19}, the AAoI of source $s_1$ in a two-source~unreliable M/G/1 queueing model can be computed as follows:\vspace{-0.3em}
\begin{equation}
	\Delta_{1\!} \!=\! \Exp[W] + 2\!\Exp[S] + \dfrac{2 \rho_{2\!} \!-\!\! 1}{\lambda_1} + \dfrac{2(\!1 \!\!-\!\! \rho_{2}\!)W^{*\!}(\lambda_{1\!})}{\lambda_1} + (\rho_2 -\! 1\!) W^{*'\!\!}(\lambda_1\!).
	\label{eq20}
	\vspace{-0.3em}
\end{equation}

In general, the AAoI of source $s_1$ in a multi-source M/G/1 model with server breakdowns can be calculated by replacing $\rho_2$ in \eqref{eq20} with $\sum_{j \in \mathcal{S}/\{1\}} \rho_j$ as follows:\vspace{-0.5em}
\begin{align}
	\Delta_{1} = &\Exp[W] + 2\Exp[S] + \dfrac{2 W^*(\lambda_1)}{\lambda_1} - W^{*'}(\lambda_1) - \dfrac{1}{\lambda_1} \nonumber \\
	&+ \sum_{j \in \mathcal{S}/\{1\}} \rho_j \left(\dfrac{2}{\lambda_j} + W^{*'}(\lambda_j) - \dfrac{2 W^*(\lambda_j)}{\lambda_j}\right).
	\label{eq21}
	\vspace{-1.9em}
\end{align}

\section{Numerical Results and Discussions}
\label{sec5}
\fontdimen2\font=0.50ex
In this section, we evaluate the AAoI in a multi-source unreliable M/G/1 queueing model (QM) and compare our analytical findings with the multi-source M/G/1 baseline model~(BL)~studied in \cite{Moltafet2020}. For this purpose, we consider Erlang-2 (Erl.) and hyper-exponential of order 2 ($H_2$) service time distributions.~The probability density function corresponding to the $H_2$ distribution is given as follows, where $\gamma_1 \!=\! 2p/m_1$, $\gamma_2 \!=\! 2(1-p)/m_1$, $p \!=\! (1 \!+ \sqrt{(c^2 \!-\! 1)/(c^2 \!+\! 1)})/2$,  $c^2 \!\geq \! 1$ is the squared coefficient of variation, and $m_1$ is the mean:\vspace{-0.4em}
\begin{align}
	f_{H_2}(t) = p \gamma_1 e^{-\gamma_1 t} + (1 - p) \gamma_2 e^{-\gamma_2 t},\quad t \!\geq\!0.
	\label{eq22}
	\vspace{-0.4em}
\end{align}

\figurename{~\ref{fig3}} shows the AAoI of source $s_1$ $(\Delta_1)$ impacted by~the packet arrival rate of $s_1$ $(\lambda_1)$ under general service time distribution for varying $N$ values. For parameters $\Exp[S] \!=\! 0.5$, $\Exp[R] \!=\! 0.3$, $\alpha \!=\! 0.1$, $p \!=\! 0.7$, and $\lambda_i (i \!\neq\! 1) \!=\! 0.12$, we observe that the waiting time of packets from all sources increases with $\lambda_1$, which in turn, causes $\Delta_1$ to increase as well. Nonetheless, for larger $\lambda_1$ values, QM results in higher $\Delta_1$ as compared to BL. This clearly reveals the impact of average server repair time, which is related to the expected service time as $\Exp[S] \!=\! \beta^{(1,1)}(1 \!+\! \alpha \Exp[R])$. We also see a substantial increase in the AAoI gap between QM and BL as $N$ goes from 2 to 4 under all three service time distributions in Fig.~3a and Fig.~3b. Such behavior is anticipated as the unreliable server becomes overloaded with packets arriving from a larger set of sources, thus further delaying the packet processing time.
\begin{figure}[t]
	\centering
	{\label{fig4b}\includegraphics[width=0.527\columnwidth]{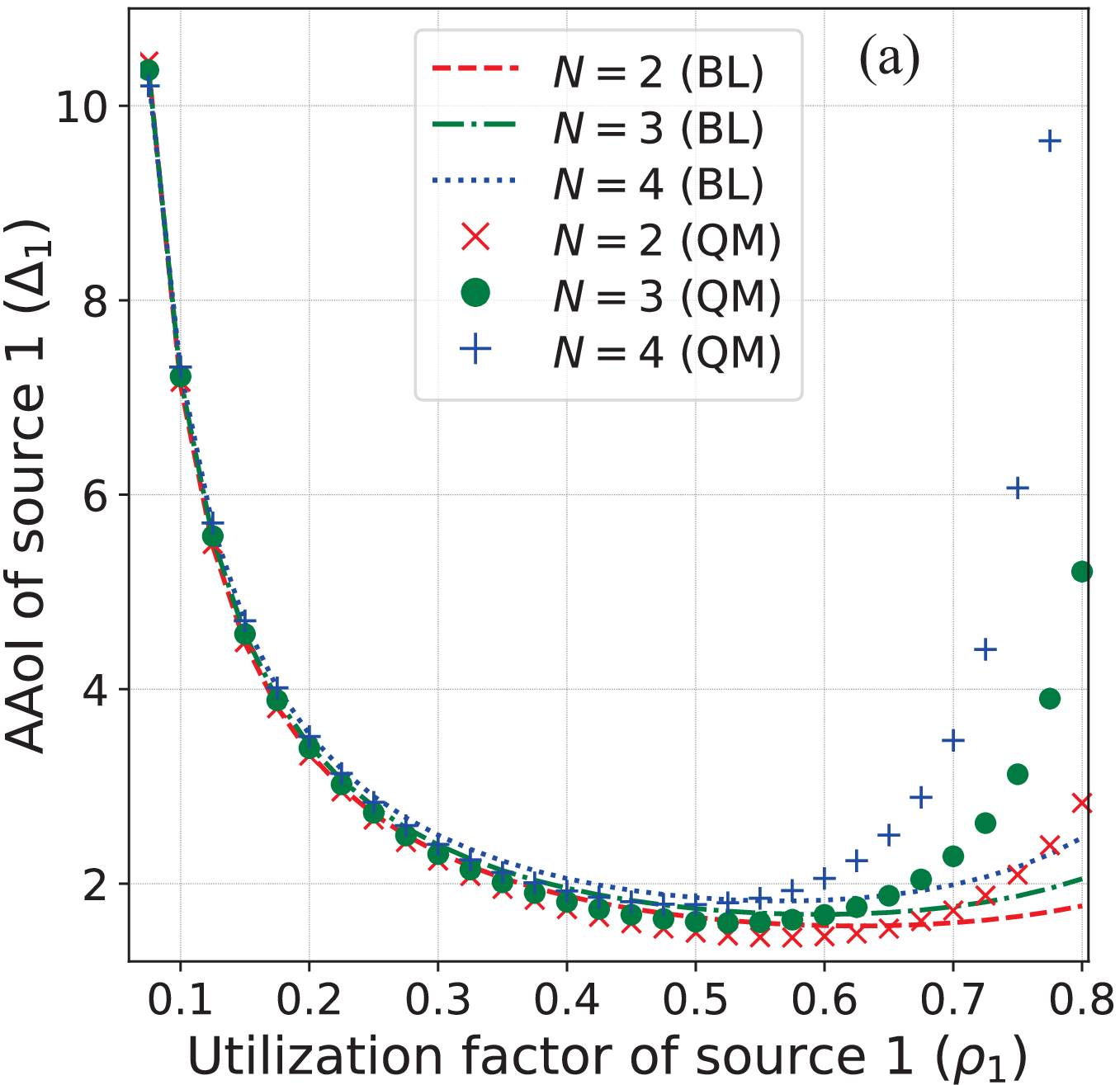}}\!
	{\label{fig4c}\includegraphics[width=0.464\columnwidth]{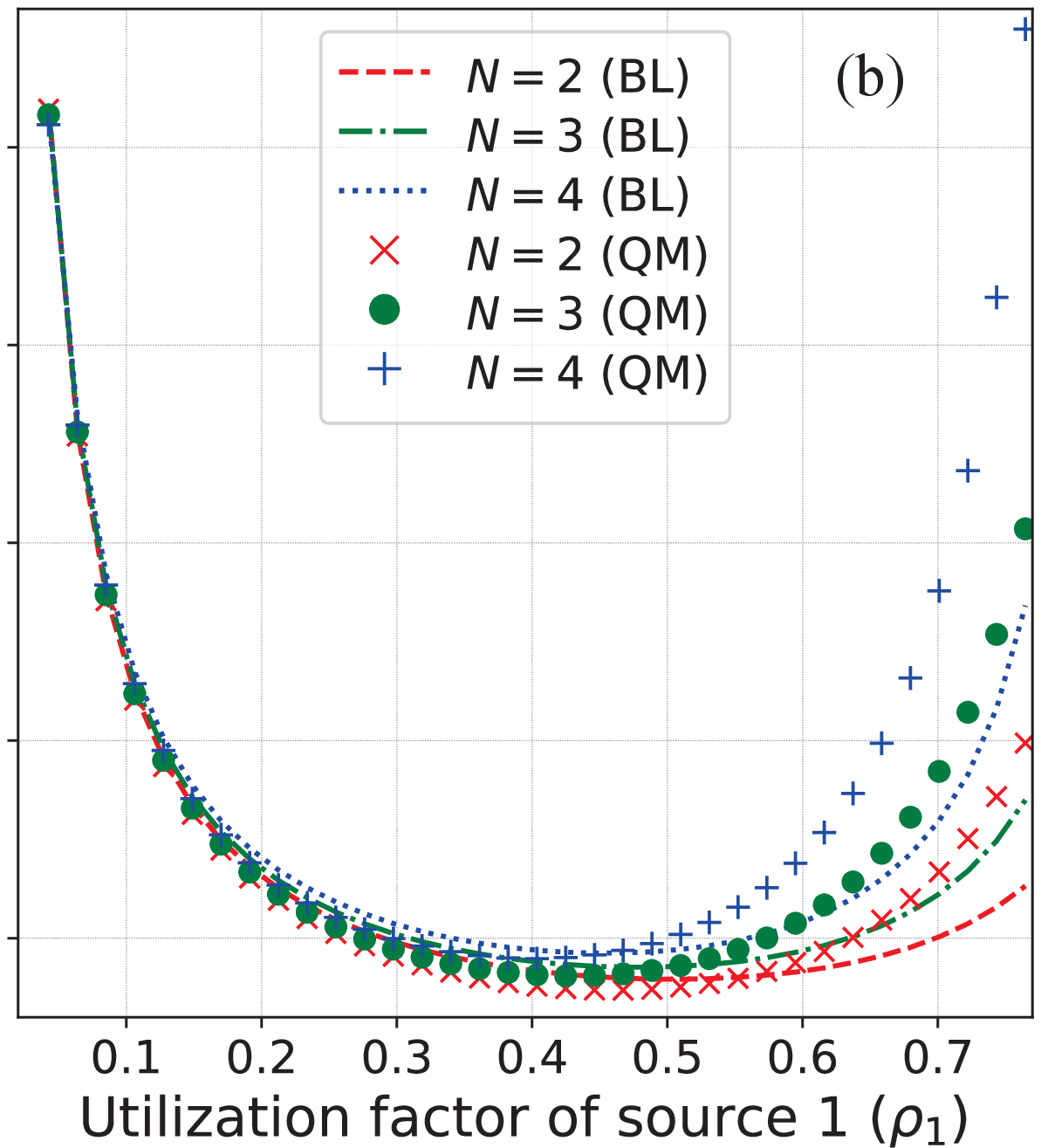}}
\vspace{-0.7em}
	\caption{The AAoI of source $s_1$ in terms of $\rho_1$ for $N \!=\! \{2,3,4\}$ under (a) Erlang-2 and (b) hyper-exponential service time distributions with $p \!=\! 0.7$. Here, $\Exp[S]=0.5$ and $\Exp[R]=0.3$, $\alpha=0.1$, and $\lambda_i (i \neq 1) \!=\! 0.12$.}
	\vspace{-0.6em}
	\label{fig4}
\end{figure}

The effect of the traffic intensity of source 1 $(\rho_1 \!=\! \lambda_1 \Exp[S])$~on $\Delta_1$ is depicted in \figurename{~\ref{fig4}}. For the same parametric settings as~in \figurename{~\ref{fig3}}, we note insignificant difference in the AAoI of $s_1$ under light traffic conditions. However, as more packets arrive at the server from multiple sources, the AAoI of $s_1$ grows non-linearly as evident in Fig.~4a and Fig.~4b. Moreover, as $\rho_1$ increases, $\Delta_1$ in QM rises at a much faster rate than BL. This is mainly because the packets being processed are interrupted by the breakdowns (and repairs) undergone by the server in QM, in spite of the low failure rate of $\alpha \!=\! 0.1$. The relative rise in $\Delta_1$ for different values of $N$ under both distributions in \figurename{~\ref{fig4}} can be clearly justified by the system dynamics illustrated in \figurename{~\ref{fig3}}.
\begin{figure}[!t]
	\centering
	\includegraphics[width=0.35\textwidth]{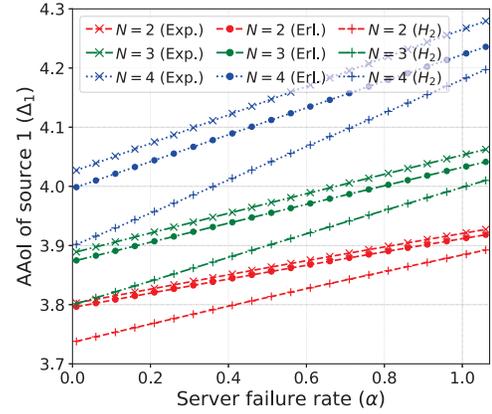}
	\vspace{-0.9em}
	\caption{The AAoI of source $s_1$ in terms of $\alpha$ for $N = \{2,3,4\}$ under various service time distributions. Here, $\Exp[S] \!=\! 0.5$, $\Exp[R] \!=\! 0.3$, $\alpha \!=\! 0.1$, $p \!=\! 0.7$, $\lambda_1 = 0.3$, and $\lambda_i (i \neq 1) \!=\! 0.12$.}
	\label{fig5}
	\vspace{-0.6em}
\end{figure}


\figurename{~\ref{fig5}} plots the AAoI of $s_1$ as a function of the server failure rate $(\alpha)$ for $\lambda_1 \!=\! 0.3$.~By comparing the service time distributions for various $N$ values, we note that the exponential distribution (Exp.), where $p \!=\!1$, yields the highest $\Delta_1$, which is then~followed by the Erlang distribution. As the server~experiences more frequent failures, the AAoI of $s_1$ under all three distributions can be seen to converge. This is merely because most of the service time~is spent repairing the server, leaving lesser time for the packets to be processed. It is also noteworthy that, irrespective of the service time distribution, QM exhibits system instability much earlier in time than BL when all sources generate packets at higher rates.

Finally, \figurename{~\ref{fig6}} showcases the AAoI of $s_1$ in terms of the mean repair time $(\Exp[R])$ and the steady-state system availability $(P_a)$. To corroborate the numerical results, we conduct Monte Carlo (MC) simulations where each data point is averaged over $1000$ runs. In \figurename{~6a}, we see that $\Delta_1$ gradually increases with $\Exp[R]$, which is conforming to the system behavior depicted in \figurename{~\ref{fig5}}. As $N$ grows, the mean repair time also increases because the server fails more frequently thus, causing the AAoI of $s_1$ to rise. As per \eqref{eq11b}, we see that $P_a$ reaches its lowest value of $0.95$ when $\lambda_1 (=0.6)$ is the highest in \figurename{~6b}. This implies that the server is more likely to process the increasing number of packets generated by $s_1$, thus resulting in lower $\Delta_1$ values. However, gradual decrease in $\lambda_1$ causes $P_a$ to rise. At $P_a \!=\! 0.995$, where $\lambda_1 \!=\! 0.06$, packets generated by the other sources are primarily served due to their higher arrival rates $(\lambda_i (i \!\neq\! 1) \!=\! 0.12)$. Such behavior explains the drastic increase of $\Delta_1$ depicted in \figurename{~6b}, despite the server availability being at its peak.\vspace{-0.2em}

\section{Conclusion}
\label{sec6}
\fontdimen2\font=0.50ex
In this paper, we analyzed the average AoI of each source~in a multi-source M/G/1 queueing
model with active server breakdowns and general repair time. Using the supplementary variable technique, we first derived the pgf of AoI and the stationary~distribution in a single-source unreliable M/G/1 queue. Then, the average AoI expression for the unreliable multi-source M/G/1 queue was derived in closed form. The analytical findings were validated for varying packet arrival and server failure rates under general service time distribution via simulation results.~A possible avenue of future work~is AoI analysis of correlated~Poisson arrivals from multiple sources in unreliable queueing models.\vspace{-0.2em}
\begin{figure}[!t]
	\centering
	{\label{fig6a}\includegraphics[width=0.512\columnwidth]{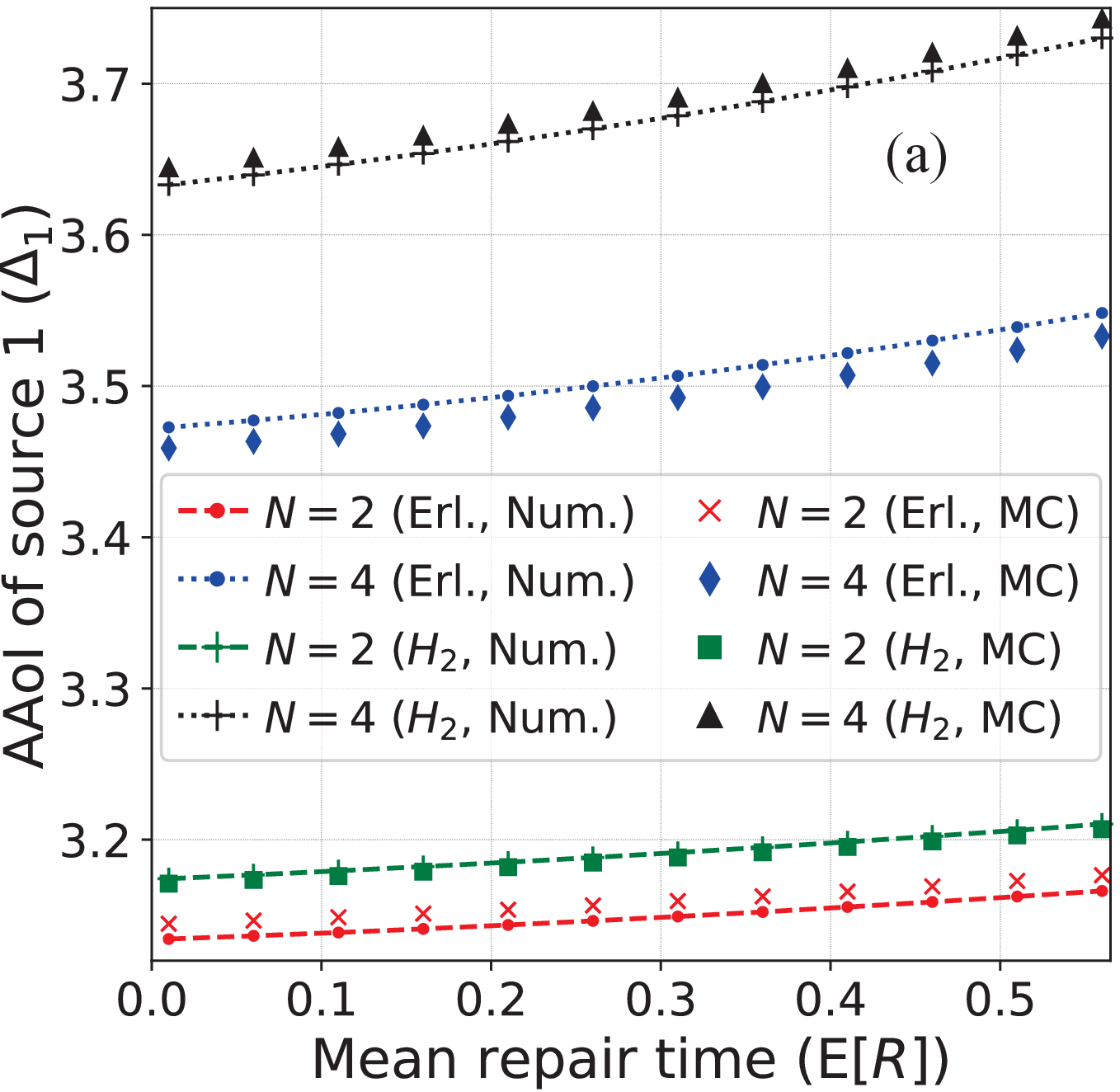}}\!
	{\label{fig6b}\includegraphics[width=0.481\columnwidth]{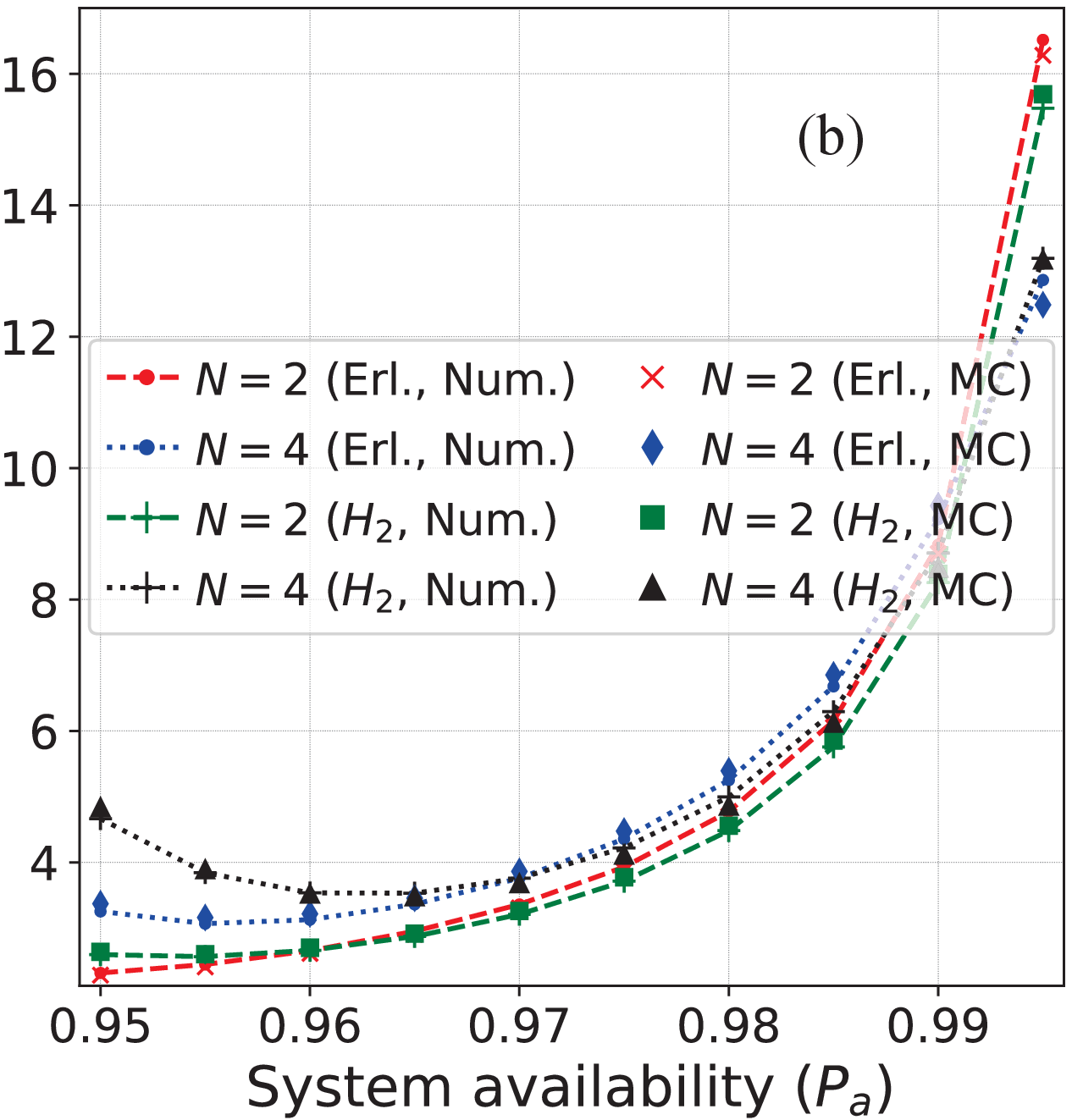}}
\vspace{-1.8em}
	\caption{The AAoI of source $s_1$ in terms of (a) $\Exp[R]$ for $\lambda_1 \!=\! 0.3$ and (b) $P_a$ for $\lambda_1 \!=\! [0.06, 0.6]$ and $\Exp[R]=0.3$. Here, $N \!=\! \{2,4\}$, $p \!=\! 0.7$, $\Exp[S]=0.5$, $\alpha=0.1$, and $\lambda_i (i \neq 1) \!=\! 0.12$.}
	\vspace{-0.5em}
	\label{fig6}
\end{figure}



{\appendices
\section{Proof of Lemma~\ref{lemma1}}
\label{appA}
The first conditional expression in \eqref{eq14} can be written as:\vspace{-0.2em}
\begin{align}
	\Exp[R^b_{1,i} X_{1,i}|A^b_{1,i}] = \Exp[T_{1,i-1} X_{1,i}|A^b_{1,i}] - \Exp[X^2_{1,i}|A^b_{1,i}].
	\label{eq23}
	\vspace{-0.2em}
\end{align}
The first term of \eqref{eq23} reduces to:\vspace{-0.2em}
\begin{align}
	\Exp[&T_{1,i-1} X_{1,i}|A^b_{1,i}] =\! \int_0^{\infty}\!\! \int_0^{\infty}\! x\,t\,f_{X_{1,i}, T_{1, i-1}| A^b_{1,i\!}}(x,t)\,\text{d} x\,\text{d} t \nonumber \\
	&= \dfrac{1}{p^b_{1,i}} \left(\dfrac{\Exp[W]}{\lambda_1} + \dfrac{\Exp[S]}{\lambda_1}- \dfrac{W^{*'\!}(\lambda_1)}{\lambda_1} - W^{*''\!}(\lambda_1)\right),
	\label{eq24}
	\vspace{-0.2em}
\end{align}
where $W^{*'\!}$ and $W^{*''\!}$ are, respectively, the first and second order derivatives of $W^*$ w.r.t. $\lambda_1$. The second term of \eqref{eq23} reduces to:\vspace{-0.2em}
\begin{align}
	\Exp[&X^2_{1,i}|A^b_{1,i}] =\! \int_0^{\infty}\! x^2\,f_{X_{1,i}| A^b_{1,i}}(x)\,\text{d} x \nonumber \\
	&= \dfrac{1}{p^b_{1,i\!}} \!\left(\!\dfrac{2}{\lambda_1^2} \!-\! \dfrac{W^{*''\!}(\lambda_1)}{\lambda_1^2} - \dfrac{2W^{*'\!}(\lambda_1^2)}{\lambda_1^2} + \dfrac{2W^{*\!}(\lambda_1)}{\lambda_1^2}\!\right).
	\label{eq25}
	\vspace{-0.2em}
\end{align}
Adding \eqref{eq24} and \eqref{eq25} gives \eqref{eq17}, which completes the proof.
\section{Proof of Lemma~\ref{lemma2}}
\label{appB}
\vspace{-0.1em}
The second conditional expression in \eqref{eq14} can be written as:\vspace{-0.3em}
\begin{align}
	\Exp[&G^b_{1,i} X_{1,i}|A^b_{1,i}] \nonumber \\ 
	&= \int_0^{\infty}\!\! x \Exp\big[\!\!\!\sum_{i' \in \Omega^b_{2,i}}\! \!\! S_{2,i'\!} \!\mid\! A^b_{1,i}, X_{1,i} \!=\! x] f_{\!X_{1\!,i}|A^b_{1,i\!\!}}(x)\,\text{d}x.
	\label{eq26}
	\vspace{-0.6em}
\end{align}
The set $\Omega^b_{2,i}$ in \eqref{eq26} denotes the packets generated from $s_2$ that must be served before those from $s_1$ and is independent of the inter-arrival time of $s_1$ packets $(T_{1,i-1})$. Thus, we get:\vspace{-0.3em}
\begin{align}
	\Exp[G^b_{1\!,i} X_{1\!,i}|A^b_{1\!,i}] \!=\! \dfrac{\rho_2}{p^b_{1\!,i}\!} \int_{0}^{\infty}\!\! x^2 \lambda_1 e^{\!-\lambda_1 x\!}\left(\!1 \!-\! F_{T_{1\!,i \!-\! 1\!}}(x)\right) \text{d}x ,
	\label{eq27}
	\vspace{-0.6em}
\end{align}
which after some algebraic manipulations and rearrangements, results in \eqref{eq18}. This completes the proof.\vspace{-0.5em}
\section{Proof of Lemma~\ref{lemma3}}
\label{appC}
\vspace{-0.1em}
The third conditional expression in \eqref{eq14} can be reduced as:\vspace{-0.4em}
\begin{align}
	 \Exp&\!\left[\left(G^l_{1\!,i} \!+\! R^l_{2,i}\right) X_{\!1\!,i}|A^l_{1\!,i}\right]  \nonumber \\
	 &\quad= \dfrac{\rho_2}{p^l_{1,i}}\! \int_0^{\infty}\!\! \int_t^{\infty}\!\! x\,t\,\lambda_1 e^{-\lambda_1 x} f_{T_{1, i-1\!}}(t)\,\text{d}x\,\text{d}t \nonumber \\
	 &\quad= \dfrac{\rho_2}{p^l_{1,i}}\! \int_0^{\infty}\!\! \left(\!t^2 e^{\!-\lambda_1 x} f_{T_{1, i - 1\!}}(t) \!+\! \dfrac{t}{\lambda_1} e^{\!-\lambda_1 x} f_{T_{1, i-1\!}}(t)\!\right) \text{d}t,
  	\label{eq28}
  	\vspace{-0.7em}
\end{align}
which after some algebraic manipulations and rearrangements, results in \eqref{eq19}. This completes the proof.\vspace{-0.2em}

\bibliographystyle{IEEEtran}
\bibliography{./IEEEabrv,./myref}

%
%
%
%
%
%
%
%
%

\end{document}